\documentclass[submission,copyright,creativecommons]{eptcs}
\providecommand{\event}{LSFA 2012} 

\usepackage{breakurl}             
\usepackage{amssymb}
\usepackage{amsmath}
\usepackage{mathrsfs}
\usepackage{stmaryrd}
\usepackage{amsthm}
\usepackage{xifthen} 
\usepackage{rotating}
\usepackage{enumerate}

\newtheorem{theorem}{Theorem}
\newtheorem{definition}[theorem]{Definition}

\newtheorem{lemma}[theorem]{Lemma}

\newtheorem{proposition}[theorem]{Proposition}

\newcommand{\bB}{\mathbf{B}}
\newcommand{\bC}{\mathbf{C}}

\newcommand{\bI}{\mathbf{I}}

\newcommand{\bR}{\mathbf{R}}
\newcommand{\bS}{\mathbf{S}}
\newcommand{\bT}{\mathbf{T}}

\newcommand{\mbbK}{\mathbb{K}}

\newcommand{\ga}{\alpha}
\newcommand{\gb}{\beta}
\newcommand{\gc}{\gamma}
\newcommand{\gd}{\delta}

\newcommand{\gs}{\sigma}
\newcommand{\gt}{\ensuremath{\tau}}

\newcommand{\go}{\omega}

\newcommand{\cA}{\mathcal{A}}

\newcommand{\cC}{\mathcal{C}}
\newcommand{\cD}{\mathcal{D}}

\newcommand{\cI}{\mathcal{I}}

\newcommand{\cK}{\mathcal{K}}

\newcommand{\cP}{\mathcal{P}}

{\begin{list}%
       {-}%
       {\setlength{\itemsep}{0pt}
     \setlength{\parsep}{3pt}
     \setlength{\topsep}{3pt}
     \setlength{\partopsep}{0pt}
     \setlength{\leftmargin}{0.7em}
     \setlength{\labelwidth}{1em}
     \setlength{\labelsep}{0.3em}}}%
{\end{list}}

{\begin{list}%
       {-}%
       {\setlength{\itemsep}{0pt}
     \setlength{\parsep}{2pt}
     \setlength{\topsep}{2pt}
     \setlength{\partopsep}{0pt}
     \setlength{\leftmargin}{2em}
     \setlength{\labelwidth}{1em}
     \setlength{\labelsep}{0.3em}}}%
{\end{list}}

\newcommand{\mbz}{\mathbf{0}}

\newcommand{\rea}[1]{\mathsf{rea}(#1)} 
\newcommand{\realize}{\Vdash} 
\newcommand{\nat}{\mathbb{N}}

\newcommand{\st}{:} 

\newcommand{\Omegatuple}[1]{\Mfin{#1}^{(\omega)}}

\newcommand{\labelto}[1]{\rightarrow_{#1}} 
\newcommand{\mslabelot}[1]{\ _{#1}\!\twoheadleftarrow} 
\newcommand{\mslabelto}[1]{\twoheadrightarrow_{#1}} 

\newcommand{\rrank}[1]{\mathsf{ar}(#1)} 
\newcommand{\termin}[3]{\mathsf{tmn}_{#1}^{#2}(#3)} 
\newcommand{\prem}[3]{\mathsf{pr}_{#1}^{#2}(#3)} 
\newcommand{\premterm}[3]{\mathsf{prt}_{#1}^{#2}(#3)} 

\newcommand{\Mfin}[1]{\mathcal{M}_{\mathrm{f}}(#1)}
\newcommand{\mcup}{\uplus}

\newcommand{\Rel}{\mathbf{Rel}} 
\newcommand{\Mon}[1]{\mathrm{Mon}_{\pa}(#1)} 
\newcommand{\CPO}{\mathbf{CPO}} 

\newcommand{\gramm}{\mathrel{::=}} 
\newcommand{\ass}{\mathrel{:=}} 

\renewcommand{\iff}{\Leftrightarrow}
\newcommand{\seq}[1]{\vec{#1}}

\newcommand{\Ps}[1]{\cP(#1)} 
\newcommand{\Pss}[1]{\cP_\mathrm{s}(#1)} 
\newcommand{\cons}{\!::\!} 
\newcommand{\at}{\!\centerdot} 
\newcommand{\ldts}{,..,} 

\newcommand{\s}{\mathsf{s}} 
\newcommand{\eo}{\eta_1} 
\newcommand{\et}{\eta_2} 

\newcommand{\cdr}[1]{\mathsf{cdr}(#1)} 
\newcommand{\car}[1]{\mathsf{car}(#1)} 
\newcommand{\cddr}[2]{#1[#2)} 
\newcommand{\cadr}[2]{#1[#2]} 

\newcommand{\nil}{\mathsf{nil}} 

\newcommand{\sub}[2]{\{#1/#2\}} 
\newcommand{\ssub}[2]{\{#1/^*#2\}} 
\newcommand{\ap}{\star} 
\newcommand{\bd}{\mu} 
\newcommand{\Te}[1]{#1^{\circ}} 

\newcommand{\Neg}[1]{#1^{-}} 
\newcommand{\Pos}[1]{#1^{+}} 

\renewcommand{\Form}{\mathsf{Fm}} 
\newcommand{\At}{\mathsf{At}} 
\newcommand{\Var}{\mathrm{Var}} 
\newcommand{\FV}{\mathrm{FV}} 
\newcommand{\FN}{\mathrm{FN}} 
\newcommand{\LTer}[1]{\Lambda^{\mathsf{#1}}} 
\newcommand{\KTer}[1]{\Sigma^{\mathsf{#1}}} 

\newcommand{\SN}[1]{\mathrm{SN}^{\mathsf{#1}}} 

\newcommand{\bbot}{
\mathrel{\vcenter{\offinterlineskip
\vskip-.130ex\hbox{\begin{turn}{90}$\models$\end{turn}}}}} 

\newcommand{\asm}{\! : \!} 

\newcommand{\tval}[1]{\vert #1\vert} 
\newcommand{\fval}[1]{\lVert #1 \rVert} 
\newcommand{\Int}[1]{\llbracket #1\rrbracket} 
\newcommand{\id}{\mathsf{id}} 

\newcommand{\copair}[2]{[ #1, #2 ]} 

\newcommand{\cur}{\Lambda} 

\newcommand{\ctrliso}{\varphi} 

\newcommand{\lmc}{\ensuremath{\lambda\mathcal{C}}} 
\newcommand{\lmu}{\ensuremath{\lambda\mu}} 
\newcommand{\ort}[1]{#1^{\bot}} 

\newcommand{\wi}{\binampersand} 
\newcommand{\pa}{\bindnasrepma} 

\newcommand{\mon}{\mathsf{m}} 

\newcommand{\dig}{\delta} 
\newcommand{\der}{\varepsilon} 

\newcommand{\alg}{\mathsf{alg}} 

\newcommand{\mul}{c} 
\newcommand{\uni}{w} 
\newcommand{\bmul}{\textbf{c}} 
\newcommand{\buni}{\textbf{w}} 
\newcommand{\comul}{d} 
\newcommand{\couni}{e} 
\newcommand{\bcomul}{\textbf{d}} 
\newcommand{\bcouni}{\textbf{e}} 

\newcommand{\teid}{\mathbf{1}} 
\newcommand{\bon}{\mathbf{1}} 

\newcommand{\foc}[1]{#1^\sharp} 

\newcommand{\varrule}{\mathsf{ax}} 
\newcommand{\absrule}[1]{\to i,{#1}} 
\newcommand{\apprule}{\to e} 
\newcommand{\namrule}{\bot i} 
\newcommand{\murule}[1]{\bot e,{#1}} 

\newcommand{\carrule}{\to e_r} 
\newcommand{\cdrrule}{\to e_l} 
\newcommand{\atrule}{\to i} 
\newcommand{\aprule}{\mathsf{cut}} 
\newcommand{\nilrule}{\bot i} 
\newcommand{\bdrule}[1]{\bd,{#1}} 

\newcommand{\leng}[1]{\sharp #1} 

\newcommand{\bang}{!} 
\newcommand{\app}{\mathrm{Ap}} 
\newcommand{\lam}{\mathrm{La}} 

\newcommand{\cln}{{\ :\ }} 
\newcommand{\semicln}{\mid} 
\newcommand{\tystk}[4]{%
\ifthenelse{\equal{#1}{s}\OR\equal{#1}{t}}{
	\ifthenelse{\equal{#1}{s}}{
		\ifthenelse{\isempty{#2}}{#3 \vdash #4}{{#2} \cln #3 \vdash #4}
		}{
		\ifthenelse{\isempty{#2}}{\vdash #3, #4}{\vdash {#2}\cln #3 \semicln #4}
		}
	}{
	\ifthenelse{\isempty{#2}}{\vdash #4}{\vdash {#2} \semicln #4}
	}
}
\newcommand{\ntystk}[4]{%
\ifthenelse{\equal{#1}{s}\OR\equal{#1}{t}}{
	\ifthenelse{\equal{#1}{s}}{
		\ifthenelse{\isempty{#2}}{#3 \nvdash #4}{{#2} \cln #3 \nvdash #4}
		}{
		\ifthenelse{\isempty{#2}}{\nvdash #3, #4}{\nvdash {#2} \cln #3 \semicln #4}
		}
	}{
	\ifthenelse{\isempty{#2}}{\nvdash #4}{\nvdash {#2} \semicln #4}
	}
}

\newcommand{\tylmu}[4]{
	\ifthenelse{\isempty{#3}}
	{
	\ifthenelse{\isempty{#4}}
		{#1 \vdash_{\lmu} {#2}}
		{#1 \vdash_{\lmu} {#2} \semicln #4}
	}
	{
        \ifthenelse{\isempty{#4}}
        	{#1 \vdash_{\lmu} {#2} \cln #3}
        	{#1 \vdash_{\lmu} {#2} \cln #3 \semicln #4}
	}
}

\newcommand{\transition}{\longrightarrow} 
\newcommand{\transitions}{-\!\!\!\!\twoheadrightarrow} 

\newcommand{\Ktclos}[2]{({#1},{#2})} 
\newcommand{\Ksclos}[2]{({#1},{#2})} 
\newcommand{\Kstate}[2]{\mbox{\boldmath{$\langle$}}{#1},{#2}\mbox{\boldmath{$\rangle$}}} 

\newcommand{\otspam}{
\mathrel{\vcenter{\offinterlineskip
\vskip-.130ex\hbox{\begin{turn}{180}$\mapsto$\end{turn}}}}} 

\newcommand{\envup}[3]{#1[#2 \otspam #3]} 

\newcommand{\lab}[2]{\mathsf{lab}_{#1}\{#2\}} 
\newcommand{\res}[2]{\mathsf{res}_{#1}\{#2\}} 

\newcommand{\deltrycatch}[3]{\mathsf{try}_{#1}\{#2\}\mathsf{catch}\{#3\}} 
\newcommand{\delthrow}[2]{\mathsf{throw}_{#1}\{#2\}} 

\input prooftree.sty

\title{The stack calculus}
\author{
Alberto Carraro
\institute{PPS, Universit\'{e} Denis Diderot Paris, France}
\email{acarraro@pps.univ-paris-diderot.fr}
\and
Thomas Ehrhard
\institute{PPS, Universit\'{e} Denis Diderot Paris, France}
\email{thomas.ehrhard@pps.univ-paris-diderot.fr}
\and
Antonino Salibra
\institute{DAIS, Universit\`{a} Ca' Foscari Venezia, Italia}
\email{salibra@dsi.unive.it}
}

\def\titlerunning{The stack calculus}
\def\authorrunning{A. Carraro, T. Ehrhard, A. Salibra}
\begin{document}
\maketitle

\begin{abstract}
We introduce a functional calculus with simple syntax and operational semantics in which the calculi introduced so far in the Curry\textendash Howard correspondence for
 Classical Logic can be faithfully encoded. Our calculus enjoys confluence without any restriction. Its type system enforces strong normalization of expressions
 and it is a sound and complete system for full implicational Classical Logic. We give a very simple denotational semantics which allows easy calculations
 of the interpretation of expressions.
\end{abstract}

\section{Introduction}

The Curry\textendash Howard correspondence \cite{Howard80} was first designed as the isomorphism between natural deduction for minimal Intuitionistic Logic \cite{Prawitz65} and the
 simply typed $\lambda$-calculus, and for a long time no one thought this isomorphism could be extended to Classical Logic, until Griffin \cite{Griffin90} proposed
 that natural deduction for Classical Logic could be viewed as a type system for a $\lambda$-calculus extended with a control operator $\cC$, introduced by Felleisen
 in his \lmc-calculus \cite{Felleisen92}. There are also other operators that correspond to logical axioms that, once added to minimal Intuitionistic Logic, give
 proof systems of different power, from minimal to full implicational Classical Logic. Felleisen's $\cC$, corresponding to the \emph{Double-Negation Elimination law},
 gives full implicational Classical Logic; less powerful operators are $\cK$ (a.k.a. $\mathsf{call/cc}$), typable with \emph{Peirce's law}, and $\cA$ (a.k.a. \emph{abort})
 typable with the \emph{Ex-Falso Quodlibet law}. On the programming side, this classification corresponds to the different expressive power of the operators as
 control primitives. Ariola and Herbelin \cite{Herbelin03} survey and classify these logical systems and introduce a refinement of \lmc-calculus which aims at
 resolving a mismatch between the operational and proof-theoretical interpretation of Felleisen's \lmc-reduction theory. 

Another extension of the $\lambda$-calculus is Parigot's \lmu-calculus \cite{Parigot92} which introduces a Natural Deduction with multiple conclusions. This system
 implements minimal Classical Logic and it is able to encode the primitive $\mathsf{call/cc}$; Ariola and Herbelin \cite{Herbelin03} extend it to cover full Classical
 Logic and compare their system with Felleisen's \lmc-calculus: similar studies are made by De Groote \cite{Groote99}. The correspondence between classical principles
 and functional control operators is further stressed by De Groote's extension of $\lambda$-calculus with \emph{raise/handle} primitives \cite{Groote95}.
 While the untyped version of \lmu-calculus enjoys confluence, its extensional version is only confluent on closed terms via the addition of a rewrite
 rule that destroys the strong normalization of typable terms \cite{David01}.

Gentzen's sequent calculus $\mathsf{LK}$ \cite{Gentzen35} is put in correspondence with a reduction system by Urban \cite{UrbanThese};
 the type system of Curien\textendash Herbelin's $\bar{\lambda}\mu\tilde{\mu}$-calculus \cite{Herbelin00}
 corresponds to its implicational fragment. These two approaches are compared in detail by Lengrand \cite{Lengrand03}.
 These calculi highlight the duality between call-by-value and call-by-name cut-elimination
 (or evaluation): confluence is not achievable without choosing one of the two strategies.
 Other computational interpretations of Classical sequent calculus are Girard's LC \cite{Gir91} and the
 translations of Classical Logic in Linear Logic \cite{DanosJS95}, based upon linear dual decomposition of classical implication.

In this paper we introduce the \emph{stack calculus}. The idea of this calculus comes from a synthesis of Krivine's extension of the $\lambda$-calculus with
 \emph{stacks} and $\mathsf{call/cc}$ \cite{Krivine01} with Parigot's \lmu-calculus. It also bears similarities with the call-by-name variant of
 $\bar{\lambda}\mu\tilde{\mu}$-calculus. In Krivine's Classical Realizability \cite{Krivine01} classical implication is associated to a stack constructor, while in
 \lmu-calculus (as in \lmc-calculus) the arrow-type is introduced by an intuitionistic $\lambda$-abstraction: the role of the $\mu$-abstraction is to make it
 classical by ``merging together'' many intuitionistic arrows. The $\mu$-abstraction can then be thought of as a functional abstraction over \emph{lists of inputs},
 corresponding to a list of consecutive $\lambda$-abstractions. This idea is used in the design of L\"{o}w\textendash Streicher's $\mathsf{CPS}_\infty$-calculus
 \cite{LowS06} which is an infinitary version of $\lambda$-calculus that allows only infinite abstractions and infinite applications.

The stack calculus is a finitary functional language in which stacks are first-class entities, and many of the previously-mentioned
 calculi can be faithfully translated. The stack calculus enjoys confluence without any restriction, also in its extensional version.
 We type the stack calculus with a propositional language with implication and falsity, to be associated to stack construction and empty stack, respectively.
 As a consequence one obtains a sound and complete system for full implicational Classical Logic.
 In our case the realizability interpretation of types \`{a} la Krivine matches perfectly the logical meaning of the arrow in the type system:
 proofs of soundness and strong normalization of the calculus are both given by particular realizability interpretations.
 The simplicity of the stack calculus, which does not use at the same time $\lambda$- and $\mu$-abstractions allows an
 easy encoding of control primitives like $\mathsf{call/cc}$, $\mathsf{label}/\mathsf{resume}$, $\mathsf{raise}/\mathsf{catch}$.

Many researchers contributed to the study of proof semantics of Classical Logic. From Girard \cite{Gir91}, to Reus and Streicher \cite{Stre98},
 to Selinger \cite{Selinger01} who gives a general presentation in terms of \emph{control categories}. It is also very interesting the work by
 Laurent and Regnier \cite{Laurent03} which shows in detail how to extract a control category out of a categorical model of
 Multiplicative Additive Linear Logic (MALL).

Inspired by Laurent and Regnier's work \cite{Laurent03} we give a minimal framework in which the stack calculus can be soundly interpreted.
 The absence of the $\lambda$-abstraction, allows us to focus on the minimal structure required to interpret Laurent's Polarized
 Linear Logic \cite{Laurent03b} and to use it to interpret the stack calculus. The simplicity of the framework gives an easy calculation of the semantics of expressions.

\section{The untyped stack calculus}\label{sec:untyped_stack}

The stack calculus has three syntactic categories: \emph{terms} that are in functional position, \emph{stacks} that are in argument position and
 represent streams of arguments, \emph{processes} that are terms applied to stacks. 
 The basis for the definition of the stack calculus language is a countably infinite set of \emph{stack variables}, ranged over by the
 initial small letters $\alpha,\beta,\gamma,\ldots$ of the greek alphabet. The language is then given by the following grammar:
\[
\begin{array}{llll}
\pi,\varpi & \gramm & \ga \mid \nil \mid M \at \pi \mid \cdr \pi & \text{stacks} \qquad\qquad\qquad\qquad\qquad  \\
M,N        & \gramm & \bd \ga.P \mid \car \pi                    & \text{terms} \\
P,Q        & \gramm & M \ap \pi                                  & \text{processes}
\end{array}
\]
We use letters $E,E'$ to range over \emph{expressions} which are either stacks, terms or processes. We denote by $\KTer{p}$, $\KTer{s}$, $\KTer{t}$, and $\KTer{e}$ the sets of all processes, stacks, terms, and expressions respectively. The operator $\bd$ is a binder. An occurrence of a variable $\ga$ in an expression $E$ is \emph{bound} if it is under the scope of a $\bd\ga$; the set $\FV(E)$ of \emph{free variables} is made of those variables having a non-bound occurrence in $E$.

\noindent \emph{Stacks} represent lists of terms: $\nil$ is the empty stack. A stack $M_1\at \cdots \at M_k \at \nil$, stands for a finite list while a stack $M_1\at \cdots \at M_k \at \ga$ stands for a non-terminated list that can be further extended.\\
\noindent \emph{Terms} are entities that wait for a stack to compute. A term $\bd\alpha.P$ is the \emph{$\bd$-abstraction} of $\alpha$ in $P$.\\
\noindent \emph{Processes} result from the \emph{application} $M\ap\pi$ of a term $M$ to a stack $\pi$. This application, unlike in $\lambda$-calculus, has to be thought as \emph{exhaustive} and gives rise to an evolving entity that does not have any outcome.

Application has precedence over $\bd$-abstraction and the stack constructor has precedence over application, so that the
 term $\bd\ga.M \ap N\at\pi$ unambiguously abbreviates $\bd\ga.(M\ap (N\at\pi))$.
 As usual, the calculus involves a substitution operator. By $E\sub{\pi}{\ga}$ we denote the (capture-avoiding) substitution of the stack $\pi$ for all free occurrences of $\ga$ in $E$.
 The symbol `$\equiv$' stands for syntactic equality, while `$\ass$' stands for definitional equality. 

\begin{lemma}[Substitution Lemma]\label{lemma:subst-Lemma}
For $E\in\KTer{e}$, $\pi,\varpi\in\KTer{s}$, $\alpha\not\in\FV(\varpi)$ and $\ga \not\equiv \gb$ we have\\
 $E\sub{\pi}{\ga}\sub{\varpi}{\gb} \equiv E\sub{\varpi}{\gb}\sub{\pi\sub{\varpi}{\gb}}{\ga}$.
\end{lemma}


\begin{definition}\label{def:reductionrules-stk}
The reduction rules of the stack calculus are the following ones:
\[
\begin{array}{ll}
(\bd)          & (\bd{\ga}.P)\ap\pi \labelto{\bd} P\sub{\pi}{\ga} \qquad\qquad\qquad\qquad\qquad\qquad\qquad\qquad\qquad \\ 
(\mathsf{car}) & \car{M\at\pi} \labelto{\mathsf{car}} M \\ 
(\mathsf{cdr}) & \cdr{M\at\pi} \labelto{\mathsf{cdr}} \pi
\end{array}
\]
Adding the following rules we obtain the \emph{extensional} stack calculus:
\[
\begin{array}{lll}
(\eo) & \bd{\ga}.M\ap\ga \labelto{\eo} M    & \text{ if } \ga \not\in \FV(M) \qquad\qquad\qquad\qquad\qquad\qquad \\ 
(\et) & \car\pi\at\cdr\pi \labelto{\et} \pi
\end{array}
\]
\end{definition}

We simply write $\labelto{\s}$ for the contextual closure of the relation $(\labelto{\bd} \cup \labelto{\mathsf{car}} \cup \labelto{\mathsf{cdr}})$.
 Moreover we write $\labelto{\eta}$ for the contextual closure of the relation $(\labelto{\eo}\cup\labelto{\et})$ and finally we set $\labelto{\s\eta} = (\labelto{\s} \cup \labelto{\eta})$. 
For example, if $\bI\ass\bd\ga.\car\ga\ap\cdr\ga$, then $\bI\ap\bI\at\nil \labelto{\s} \bI\ap\nil \labelto{\s} \car\nil\ap\cdr\nil$ and the reduction 
 does not proceed further. If $\go \ass \bd \ga.\car{\ga}\ap\ga$, then $\go\ap\go\at\nil \labelto{\s} \go\ap\go\at\nil$; this
 is an example of a non-normalizing process. The stack calculus enjoys confluence, even in its extensional version, as the following theorems state.

\begin{theorem}\label{thm:CR-bd}
The $\labelto{\s}$-reduction is Church-Rosser.
\end{theorem}

\begin{theorem}\label{thm:CR-bd-eta}
The $\labelto{\s\eta}$-reduction is Church-Rosser.
\end{theorem}

We observe that Theorem \ref{thm:CR-bd-eta} holds despite the non left-linearity of the reduction rules of the extensional stack calculus.
 In other calculi, like the $\lambda$-calculus with surjective pairing, the interaction of the extensionality rule with the projection rules breaks the
 Church-Rosser property for the calculus \cite{Klop89}.

\subsection{Translation of lambda-mu-calculus}\label{subsec:tran-lmu}


Many calculi have been introduced so far to extend the Curry\textendash Howard correspondence to classical logic \cite{Griffin90,Parigot92,Groote95,UrbanThese,Herbelin00}.
 Since we cannot attempt to report a comparison with the stack calculus for each one of them, so we choose probably the best known, i.e. Parigot's
 \lmu-calculus. In this section we show how \lmu-calculus can be faithfully encoded into the stack calculus (in the precise sense of the forthcoming
 Theorem \ref{thm:preserv-red-lmu}).

The basis for the definition of the \lmu-calculus language are two (disjoint) sets $\lambda\Var$ and $\mu\Var$ of $\lambda$-variables and $\mu$-variables (a.k.a. \emph{names}), respectively. The names, ranged over by $\ga,\gb,\gc,\ldots$, are are taken from $\mu\Var$ and the usual variables, taken in $\lambda\Var$, are ranged over by $x,y,z,\ldots$.
 The expressions belonging to the language of \lmu-calculus are often divided into two categories, \emph{terms} and \emph{named terms}, produced by the following grammar:
\[
\begin{array}{llll}
s,t & \gramm & x \mid \lambda x.t \mid st \mid \mu \ga.p & \text{terms} \qquad\qquad\qquad\qquad\qquad \\
p,q & \gramm & [\alpha]t                                 & \text{named terms}
\end{array}
\]
We use letters $e,e'$ to range over \emph{expressions} which are either terms or named terms. 
We denote by $\LTer{t}$, $\LTer{p}$, and $\LTer{e}$ the sets of all terms, named terms and expressions, respectively.

We briefly recall the operational semantics of \lmu-calculus. In addition to the usual capture-free substitution $e\sub{t}{x}$ of a term $t$ for a variable $x$ in $e$, \lmu-calculus uses the \emph{renaming} $e\sub{\gb}{\ga}$ of $\ga$ with $\gb$ in $e$ and the \emph{structural substitution} $e\ssub{s}{\ga}$ that replaces all named subterms $[\ga]t$ of $e$ with the named term $[\ga]ts$: for example $(\lambda y.\mu\gb.[\ga]z)\ssub{\lambda x.x}{\ga} \equiv \lambda y.\mu\gb.[\ga]z(\lambda x.x)$ (see \cite{Parigot92}).
 Note that we adopt here the notations of David and Py \cite{David01} instead of Parigot's original ones.
 The reduction relation characterizing the \lmu-calculus is given by the contextual closure of the following rewrite rules:
\[
\begin{array}{llllll}
(\beta) & (\lambda x.t)s \labelto{\beta} t\sub{s}{x}              & \text{ logical reduction }
 &
(\rho)  & [\beta](\mu \ga.p) \labelto{\rho} p\sub{\beta}{\ga}     & \text{ renaming } \\
(\mu)   & (\mu \alpha.p)s \labelto{\mu} \mu \alpha.p\ssub{s}{\ga} & \text{ structural reduction }
 & 
(\theta)& \mu \alpha.[\alpha]t \labelto{\theta} t               & \text{ if } \alpha \not\in \FN(t)
\end{array}
\]
The reduction $\labelto{\beta\mu\rho\theta}$ was proved to enjoy the Church-Rosser property by Parigot \cite{Parigot92}. The extensional \lmu-calculus is obtained by adding the contextual closure of the following reduction rules:
\[
\begin{array}{lll}
(\eta) & \lambda x.tx \labelto{\eta} t                                & \text{ if } x \not\in \FV(t) \qquad\qquad\qquad\qquad\qquad\qquad \\
(\nu)  & \mu\alpha.p \labelto{\nu} \lambda x.\mu\alpha.p\ssub{x}{\ga} & \text{ if } x \not\in \FV(p)
\end{array}
\]
We are now going to translate \lmu-expressions into expressions of the stack calculus (stack-expressions, for short).
 A minor technical detail for the translation is the need of regarding all $\lambda$-variables and all names as stack variables.

\begin{definition}\label{def:translation-stk-lmu}
Define a mapping $\Te{(\cdot)}: \LTer{e} \to \KTer{e}$ by induction as follows:
\[
\begin{array}{ll}
\Te{x} = \bd\gb.\car{x}\ap\gb                  &                                          \\ 
\Te{(\lambda x.t)} = \bd{x}.\Te{t} \ap \cdr{x} &                                          \\
\Te{(ts)} = \bd\gb.\Te{t}\ap \Te{s}\at\gb      & \gb \not\in \FV(\Te{t}) \cup \FV(\Te{s}) \qquad\qquad\qquad\qquad\qquad\qquad\qquad\qquad\qquad \\
\Te{([\alpha]t)} = \Te{t} \ap \alpha           &                                          \\
\Te{(\mu \alpha.p)} = \bd{\alpha}.\Te{p}       &
\end{array}
\]
\end{definition}

The translation of Definition \ref{def:translation-stk-lmu} preserves the convertibility of expressions and in this sense provides an embedding of \lmu-calculus
 into the stack calculus.

\begin{theorem}\label{thm:preserv-red-lmu}
Let $e,e' \in \LTer{e}$. 
\begin{enumerate}[(i)]
\item If $e \labelto{\beta\mu\rho\theta} e'$, then $\Te{e}$ and $\Te{(e')}$ have a common reduct in the stack calculus.
\item If $e \labelto{\beta\mu\rho\theta\eta\nu} e'$, then $\Te{e}$ and $\Te{(e')}$ have a common reduct in the extensional stack calculus.
\end{enumerate}
\end{theorem}

Note that the extensional \lmu-calculus does not enjoy a full Church-Rosser theorem, as witnessed by the following counterexample \cite{David01}:
 $[\gamma]y \mslabelot{\eta\rho} [\beta]\lambda x.(\mu\ga.[\gamma]y)x \labelto{\mu} [\beta]\lambda x.\mu\ga.[\gamma]y$.

However these kinds of situations do not arise in the stack calculus (by Theorem \ref{thm:CR-bd-eta}): in this case for example we have
 $\Te{([\gamma]y)} \mslabelto{\s} \car{y}\ap\gamma \mslabelot{\s\eta} \Te{([\beta]\lambda x.\mu\ga.[\gamma]y)}$.

For example $\Te{(\lambda x.x)} = \bd x.\car{x}\ap\cdr{x}$ and
 $\Te{(\mathsf{call}/\mathsf{cc})} = \mu\ga.\car{\ga}\ap(\mu\gb.\car{\gb}\ap\cdr{\ga})\at\cdr{\ga}$,
 where $\mathsf{call}/\mathsf{cc} \equiv \lambda f.\mu\ga.[\ga](f(\lambda x.\mu\delta.[\ga]x))$.


\section{The typed stack calculus}\label{sec:typed-stack}

We are now going to look at the stack calculus in the light of the Curry\textendash Howard isomorphism. Since the stack calculus can encode calculi with control
 features (such as \lmu-calculus), it can be given a deductive system of full classical implicational propositional logic ($\{\to,\bot\}$-fragment).

The type system has judgements that come in three forms: $\tystk{s}{\pi}{A}{\Delta}$, $\tystk{t}{M}{A}{\Delta}$, and $\tystk{p}{P}{}{\Delta}$, where
 as usual greek capital letters $\Delta,\Delta'$ are used to denote \emph{contexts}, that is sets of
 assumptions $\{\ga_1\asm A_1,\ldots,\ga_n\asm A_n\}$ (also abbreviated by $\seq \ga\asm\seq A$). In a judgement like $\tystk{t}{M}{A}{\Delta}$, the 
 vertical bar separates the context $\Delta$ from the \emph{active formula} $A$; Theorem \ref{thm:lmu-tranlsation-type} can sharpen its role
 via a comparison with judgements in typed $\lambda\mu$-calculus.

\noindent
\begin{tabular}{|cccc|}
\hline
 & & & \\
\begin{prooftree}
\tystk{t}{M}{A}{\Delta} \qquad \tystk{s}{\pi}{B}{\Delta}
\justifies
\tystk{s}{M\at\pi}{A \to B}{\Delta}
\thickness=0.08em
\using{[\atrule]}
\end{prooftree}
 &
\begin{prooftree}
{\ga}\asm{A} \in \Delta
\justifies
\tystk{s}{\ga}{A}{\Delta}
\thickness=0.08em
\using{[\varrule]}
\end{prooftree}
 & 
\begin{prooftree}
\tystk{s}{\pi}{A \to B}{\Delta}
\justifies
\tystk{s}{\cdr\pi}{B}{\Delta}
\thickness=0.08em
\using{[\cdrrule]}
\end{prooftree}
 &
 \\ & & & \\
\begin{prooftree}
\tystk{s}{\pi}{A\to B}{\Delta}
\justifies
\tystk{t}{\car\pi}{A}{\Delta}
\thickness=0.08em
\using{[\carrule]}
\end{prooftree}
 &
\begin{prooftree}
\tystk{p}{P}{}{\Delta,\ga\asm A}
\justifies
\tystk{t}{\bd \ga.P}{A}{\Delta}
\thickness=0.08em
\using{[\bdrule\ga]}
\end{prooftree}
 &
\begin{prooftree}
\tystk{t}{M}{A}{\Delta} \qquad \tystk{s}{\pi}{A}{\Delta}
\justifies
\tystk{p}{M \ap \pi}{}{\Delta}
\thickness=0.08em
\using{[\aprule]}
\end{prooftree}
 &
\\
& & & \\
\begin{prooftree}
\ 
\justifies
\tystk{s}{\nil}{\bot}{\Delta}
\thickness=0.08em
\using{[\nilrule]}
\end{prooftree}
 & & & \\
 & & & \\
\hline
\end{tabular}\\

\textbf{Fig 2:} Typed stack calculus - propositional $\{\to,\bot\}$-fragment.\\


The choice for the forms of the judgements is justified by the forthcoming Theorem \ref{thm:lmu-tranlsation-type}, where it will appear that the role
 of contexts is analogous to that of \emph{name contexts} (i.e. right contexts) in typed \lmu-calculus (see Figure 3).

It is very well-known that by restricting Gentzen's sequent calculus $\mathsf{LK}$ \cite{Gentzen35} to manage at most one formula on the right-hand side of sequents one gets the
 intuitionistic sequent calculus. On the other hand, the symmetric restriction (which, by symmetry, is well behaved with respect to cut elimination)
 is not so popular. One can find an explicit study of the induced system in Czermak \cite{Cze77}. In \cite{Laurent11} Laurent studies a slight 
 variation of Czermak's system, that he calls $\mathsf{LD}_0$, and explores the logical duality between $\mathsf{LD}_0$ and its symmetrical
 calculus $\mathsf{LJ}_0$. The existence of these two symmetrical (and equivalent, via duality) systems has its roots in the dual 
 ``decomposition'' of $\mathsf{LK}$ into Danos et. al's \cite{DanosJS95} $\mathsf{LKQ}$ and $\mathsf{LKT}$ systems, corresponding to call-by-value
 and call-by-name evaluation of classical proofs, respectively. Both systems are as powerful as $\mathsf{LK}$, and $\mathsf{LKT}$ can be encoded
 into $\mathsf{LD}_0$, in which the \emph{stoup} disappears, since there is at most one formula on the left-hand side of sequents. There is a close
 relationship between $\mathsf{LD}_0$, $\mathsf{LKT}$ and the stack calculus, but indeed while the first two are formulated
 as a sequent calculus (i.e., with introduction rules only) the latter has elimination rules. One can translate both $\mathsf{LKT}$ and $\mathsf{LD}_0$
 into the stack calculus (and viceversa), somewhat as Gentzen's $\mathsf{LK}$ can be translated into Prawitz's natural decuction \cite{Prawitz65}
 (and viceversa) but the translations are not mere inclusions.


The judgements in stack calculus have the following intuitive logical interpretation, in terms of the classical (boolean) notion of semantic entailment
 ``$\vDash$". For those of the form $\tystk{s}{\pi}{A}{\gb_1\asm B_1,\ldots,\gb_n\asm B_n}$, read ``$\neg B_1,\ldots,\neg B_n \vDash \neg A$"; for
 those of the form $\tystk{t}{M}{A}{\gb_1\asm B_1,\ldots,\gb_n\asm B_n}$, read ``$\neg B_1,\ldots,\neg B_n \vDash A$"; for those of the form
 $\tystk{p}{P}{}{\gb_1\asm B_1,\ldots,\gb_n\asm B_n}$, read ``$\neg B_1,\ldots,\neg B_n \vDash \bot$". The above indications will be restated and proved
 precisely in Theorem \ref{thm:soundness}.

We now show that the reduction rules specified in Section \ref{sec:untyped_stack} are indeed reduction rules for the proofs of the typed system.

\begin{lemma}[Typed substitution lemma]\label{lem:typed-subs}
Suppose $\tystk{s}{\pi}{B}{\Delta}$.
\begin{itemize}
 \item[(i)] If $\tystk{s}{\varpi}{A}{\gb \asm B,\Delta}$, then $\tystk{s}{\varpi\sub{\pi}{\gb}}{A}{\Delta}$
 \item[(ii)] if $\tystk{t}{M}{A}{\gb \asm B,\Delta}$, then $\tystk{t}{M\sub{\pi}{\gb}}{A}{\Delta}$
 \item[(iii)] if $\tystk{p}{P}{}{\gb \asm B,\Delta}$, then $\tystk{p}{P\sub{\pi}{\gb}}{}{\Delta}$.
\end{itemize}
\end{lemma}

Using Lemma \ref{lem:typed-subs}, we can prove that the reduction of a typed term preserves the type.

\begin{theorem}\label{thm:subj-red-stk}
For all $\pi,\pi' \in \KTer{s}$, all $P,P' \in \KTer{p}$ and $M,M' \in \KTer{t}$ we have that
\begin{itemize}
 \item[(i)] if $\tystk{p}{P}{}{\Delta}$ and $P \labelto{\s\eta} P'$, then $\tystk{p}{P'}{}{\Delta}$
 \item[(ii)] if $\tystk{s}{\pi}{A}{\Delta}$ and $\pi \labelto{\s\eta} \pi'$, then $\tystk{s}{\pi'}{A}{\Delta}$
 \item[(iii)] if $\tystk{t}{M}{A}{\Delta}$ and $M \labelto{\s\eta} M'$, then $\tystk{t}{M'}{A}{\Delta}$.
\end{itemize}
\end{theorem}

Another way to type the stack calculus is to choose a language with negation, conjunction and falsity, to be associated to abstraction,
 stack construction and empty stack, respectively. This approach mirrors the one used by Lafont et al. \cite{Stre93} to type 
 the $\lambda$-calculus with explicit pair constructor and projections. The result is an intuitionistic proof system that can be seen as the target of a CPS translation that embeds Classical Logic into a fragment of
 Intuitionistic Logic via a mapping that transforms the types but not the proofs; this can be done by two translations $\Pos{(\cdot)}$
 and $\Neg{(\cdot)}$ from $\{\to,\bot\}$-formulas into $\{\wedge,\neg,\bot\}$-formulas as follows:
 $\Pos{\bot} = \neg\bot$ and $\Pos{a} = a$, for every atom $a$; $\Pos{(A \to B)} = \Neg{A} \wedge \Pos{B}$; $\Neg{A} = \neg \Pos{A}$.
One obtains a ``rule-per-rule'' correspondence: under this point of view, the stack calculus is the target-language of a CPS translation from itself that alters the types but not the proofs, while
 the translation of Lafont et al. does change also the terms. 
\subsection{Translation of typed lambda-mu-calculus}\label{subsec:tran-lmu-typed}

The \lmu-calculus is endowed with a type system that is a sound and complete Natural Deduction system for purely implicational classical logic.

The type system has judgements that come in two forms: $\tylmu{\Gamma}{t}{A}{\Delta}$ and $\tylmu{\Gamma}{p}{}{\Delta}$. On the left-hand side, $\Gamma$ represents a context $\seq x\asm\seq A$ of assumptions for the free $\lambda$-variables, while on the right-and side, $\Delta$ represents a context $\seq \ga\asm\seq B$ of assumptions for the free names.

\noindent
\begin{tabular}{|ccc|}
\hline
 & & \\
\begin{prooftree}
\tylmu{\Gamma,x\asm A}{t}{B}{\Delta}
\justifies
\tylmu{\Gamma}{\lambda x.t}{A \to B}{\Delta}
\thickness=0.08em
\using{[\absrule{x}]}
\end{prooftree}
 &
\begin{prooftree}
\tylmu{\Gamma}{t}{A\to B}{\Delta} \quad \tylmu{\Gamma}{s}{A}{\Delta}
\justifies
\tylmu{\Gamma}{ts}{B}{\Delta}
\thickness=0.08em
\using{[\apprule]}
\end{prooftree}
&
 \\ & & \\
\begin{prooftree}
\tylmu{\Gamma}{t}{A}{\Delta}
\justifies
\tylmu{\Gamma}{[\alpha]t}{}{\ga \asm A,\Delta}
\thickness=0.08em
\using{[\namrule]}
\end{prooftree}
 &
\begin{prooftree}
\tylmu{\Gamma}{p}{}{\beta\asm B,\Delta}
\justifies
\tylmu{\Gamma}{\mu\beta.p}{B}{\Delta}
\thickness=0.08em
\using{[\murule{\beta}]}
\end{prooftree}
 &
\begin{prooftree}
x \asm A \in \Gamma
\justifies
\tylmu{\Gamma}{x}{A}{\Delta}
\thickness=0.08em
\using{[\varrule]}
\end{prooftree}
 \\ & & \\
\hline
\end{tabular}\\

\textbf{Fig. 3:} Typed \lmu-calculus - propositional $\{\to\}$-fragment.


Given a context $\Gamma = x_1\asm A_1,\ldots,x_n\asm A_n$ and a sequence of formulas $\seq C = C_1,\ldots,C_n$ we write $\Gamma \to \seq C$ as an abbreviation for $x_1\asm A_1\to C_1,\ldots,x_n\asm A_n\to C_n$.

\begin{theorem}\label{thm:lmu-tranlsation-type}
\begin{enumerate}[(i)]
\item If $\tylmu{\Gamma}{t}{B}{\Delta}$, then for all sequences $\seq C$ of formulas we have $\tystk{t}{\Te{t}}{B}{\Gamma \to \seq C,\Delta}$.
\item If $\tylmu{\Gamma}{p}{}{\Delta}$, then for all sequences $\seq C$ of formulas we have $\tystk{p}{\Te{p}}{}{\Gamma \to \seq C,\Delta}$.
\end{enumerate}
\end{theorem}

From Theorem \ref{thm:lmu-tranlsation-type} results clearly that when the $\lambda$-variables are looked at as stack variables, they are endowed with a stream type of which only the type of the head is uniquely determined.

Finally we observe that the empty stack $\nil$ does not appear in the translations of \lmu-terms. It is needed if we want to translate the so-called \lmu-top calculus \cite{Herbelin03}: in fact one can naturally set $\Te{([top]t)} = \Te{t} \ap \nil$.

\subsection{Realizability interpretation of classical logic via stack calculus}\label{subsec:real-stack}

In this section we set up a framework which is the analogue of Krivine's Classical Realizability \cite{Krivine01}. Krivine's idea is to interpret
 implicational formulas at the same time as sets of stacks and sets of terms of his modified $\lambda$-calculus obtaining, respectively, falsehood and
 truth values for the formulas. This method has many applications, among which the extraction of programs \emph{realizing} mathematical theorems in the
 context of relevant logical theories such as Zermelo\textendash Frenkel Set Theory and Analysis \cite{Krivine01}. We will apply particular
 instances of realizability interpretation in Sections \ref{sec:sound-compl-stack} and \ref{subsec:norm-tystk} to prove soundness and strong
 normalization of our typed calculus.

Let $\bT\subseteq \KTer{t}$ and $\mbz\subseteq \KTer{s}$ be given sets of terms and stacks, respectively,
 such that $\nil \in \mbz$ and \\ if $M \in \bT$ and $\pi \in \mbz$, then $M\at\pi\in \mbz$ and $\cdr{\pi}\in \mbz$.

We define three binary relations $\succ_{\mathsf{s}}$, $\succ_{\mathsf{t}}$, $\succ_{\mathsf{p}}$ on
 $\KTer{s}$, $\KTer{t}$ and $\KTer{p}$, respectively, as the smallest reflexive relations satisfying the following conditions:
\begin{itemize}
\item $\succ_{\mathsf{s}}$ is transitive;
\item if $M\in\bT$, $\pi\in\mbz$ and $\varpi \succ_{\mathsf{s}} M\at\pi$, then $\car{\varpi} \succ_{\mathsf{t}} M$ and $\cdr{\varpi} \succ_{\mathsf{s}} \pi$;
\item if $\pi\in\mbz$, then $(\bd\ga.P)\ap \pi \succ_{\mathsf{p}} P\sub{\pi}{\ga}$;
\item if $M' \succ_{\mathsf{t}} M$, then $M'\ap\pi \succ_{\mathsf{p}} M\ap\pi$.
\end{itemize}

Moreover we let $\succ_{\mathsf{e}} = \succ_{\mathsf{p}} \cup \succ_{\mathsf{s}} \cup \succ_{\mathsf{t}}$
 and we say that a set $X\subseteq \KTer{e}$ is \emph{saturated} if $E\in X$ and $E' \succ_{\mathsf{e}} E$ imply $E' \in X$.
 For $X\subseteq \KTer{e}$, we let $\Pss{X}$ denote the family of all saturated subsets of $X$.

\begin{definition}\label{def:realizability-triple}
A triple $(\bbot,\bT,\mbz)$ of sets is a \emph{realizability triple} if $\bbot\subseteq \KTer{p}$, $\bT\subseteq \KTer{t}$, $\mbz\subseteq \KTer{s}$
 are all saturated.
\end{definition}

\begin{definition}[Realizability relation]\label{def:realizes}
Let $(\bbot,\bT,\mbz)$ be a realizability triple. We define a binary relation\\
 $\realize \subseteq \bT \times \Pss{\mbz}$ as $M \realize X$ iff $\forall \pi \in X.\ M \ap \pi \in \bbot$.
\end{definition}

If $M \realize X$, we say that $M$ \emph{realizes} $X$, or that $M$ is a \emph{realizer} of $X$; the set of realizers of $X$ is
 \mbox{$\rea{X} = \{ M \in \bT \st M \realize X\}$}. We define the following binary operation on $\Ps{\KTer{s}}$ as follows:\\
 \mbox{$X \Rightarrow Y= \{\varpi \in \mbz \st \exists M \in \rea{X}.\exists \pi \in Y.\ \varpi \succ_{\mathsf{s}} M\at\pi\}$}.

We indicate by $\At$ the set of all atomic formulas, which includes $\bot$ and a countable set of atoms.
 We indicate by $\Form$ the set of all formulas built from $\At$ with the connective $\to$.
 We use the following conventions: letters $A,B,C,\ldots$ range over $\Form$, and $F,G,H,\ldots$ range over $\At$.
 We let arrows associate to the right, so that $A \to B \to C \equiv A \to (B \to C)$. Every formula is of the form
 $B_1 \to \cdots \to B_n \to G$, where $G$ is atomic. As usual the negation is defined as $\neg A \ass A \to \bot$.
 
Let $\bR=(\bbot,\bT,\mbz)$ be a realizability triple. An \emph{atomic $\bR$-interpretation} is a function $\cI: \At \to \Pss{\mbz}$ such that
 $\cI(\bot) = \mbz$. Then $\cI$ extends uniquely to a map $\fval{\cdot}_{\cI}: \Form \to \Ps{\KTer{s}}$ by setting
 $\fval{A \to B}_{\cI} = \fval{A}_{\cI} \Rightarrow \fval{B}_{\cI}$. The set $\fval{A}_{\cI}$ is called the \emph{falsehood value} of the formula $A$ under $\cI$.
 The \emph{truth value} $\tval{A}_{\cI}$ of a formula $A$ under $\cI$ is given by $\tval{A}_{\cI} = \rea{\fval{A}_{\cI}}$.

\begin{proposition}\label{prop:sat-int}
For every formula $A$, $\fval{A}_{\cI} \in \Pss{\mbz}$ and $\tval{A}_{\cI} \in \Pss{\bT}$.
\end{proposition}

\begin{proof}
By induction on the structure of formulas. For falsehood values is suffices
 to observe that $\Pss{\mbz}$ is closed under the $\Rightarrow$ operation. For truth values,
 use the fact that $M' \succ_{\mathsf{t}} M$ implies $M'\ap\pi \succ_{\mathsf{p}} M\ap\pi$
 and the saturation of $\bbot$.
\end{proof}

If $\seq \pi = \pi_1,\ldots,\pi_n$ and $\seq B = B_1,\ldots,B_n$ are sequences, we write $\seq \pi \in \fval{\seq B}_{\cI}$ as an abbreviation for $\pi_1 \in \fval{B_1}_{\cI},\ldots, \pi_n \in \fval{B_n}_{\cI}$. The next theorem 
is the stack calculus analogue of Krivine's Adequacy Theorem \cite{Krivine01}, which shows that realizability is compatible with deduction in classical logic. It is an essential tool that will be used to obtain, in a uniform way, both soundness and strong normalization of the typed calculus.

\begin{theorem}[Adequacy theorem]\label{thm:adequacy}
Let $\bR =(\bbot,\bT,\mbz)$ be a realizability triple and let $\cI$ be an $\bR$-interpretation.
 If $\seq \pi \in \fval{\seq B}_{\cI}$ then
\begin{itemize}
\item[(i)] If $\tystk{s}{\varpi}{A}{\seq \ga \asm \seq B}$, then $\varpi\sub{\seq \pi}{\seq \ga} \in \fval{A}_{\cI}$;
\item[(ii)] If $\tystk{t}{M}{A}{\seq \ga \asm \seq B}$, then $M\sub{\seq \pi}{\seq \ga} \in \tval{A}_{\cI}$;
\item[(iii)] If $\tystk{p}{P}{}{\seq \ga \asm \seq B}$, then $P\sub{\seq \pi}{\seq \ga} \in \bbot$.
\end{itemize}
\end{theorem}

One proves all items simultaneously proceeding by induction on the depth of type derivations.
\subsection{Normalization in the typed stack calculus}\label{subsec:norm-tystk}

We are now going to prove that the typed stack calculus is strongly normalizing. We prove this fact by adapting the reducibility candidates technique to our
 setting. It becomes a sort of instance of Krivine's adequacy theorem in the context of Classical Realizability. We let $\SN{e} \subseteq \KTer{e}$ be the set
 of all strongly normalizing expressions of the stack calculus (w.r.t. $\labelto{\s\eta}$-reduction);
 $\SN{t}$, $\SN{p}$, $\SN{s}$ denote the sets all strongly normalizing terms, processes and stacks, respectively.
 
\begin{proposition}\label{prop:sn-triple}
$\bS = (\SN{p},\SN{t},\SN{s})$ is a realizability triple. 
\end{proposition}

The proof of Proposition \ref{prop:sn-triple} consists in showing that if $E' \succ_{\mathsf{e}} E$ and $E\in\SN{p}$ (resp. $E\in\SN{t}$, $E\in\SN{s}$), then
 also $E'\in\SN{p}$ (resp. $E'\in\SN{t}$, $E'\in\SN{s}$). One can proceed by induction on the definition of $\succ_{\mathsf{e}}$. The main point of such a proof is   
 when we consider the case in which $P\equiv M\ap\pi\in\SN{p}$ and $P'\equiv M'\ap\pi$ with $M'\succ_{\mathsf{t}} M$ because there there exist $\varpi$ and $\pi'$ such that
 $\varpi \succ_{\mathsf{s}} M\at\pi'$ and $M'\equiv \car{\varpi}$. Here one can show that if $M'\ap\pi$ has an infinite reduction path, then $M\ap\pi$ has an infinite reduction path too.
 Note that it is crucial that for the terms $M' \equiv \bd\ga.(\bd\gb.\cadr{\gb}{1}\ap\gb)\ap(\bd\gc.\cadr{\ga}{0}\ap\ga)\at\nil$ and $M \equiv \bd\ga.\cadr{\nil}{0}\ap\nil$ we have
 $M'\not\succ_{\mathsf{t}} M$. In fact, setting $\pi \equiv (\bd\gd.\cadr{\gd}{0}\ap\gd)\at\nil$, we obtain that $M\ap\pi$ is strongly normalizing but $M'\ap\pi$ is not strongly normalizing.

Let $A$ be a formula. We define its \emph{arity} $\rrank{A}$ by induction setting $\rrank{G} = 0$ and\\ $\rrank{A\to B} = 1 + \rrank{B}$.
 It is convenient sometimes to use abbreviations $\cddr{\pi}{n} \ass \mathsf{cdr}(\cdots\mathsf{cdr}(\pi)\cdots)$\\
 ($n$ times) and $\cadr{\pi}{n} \ass \car{\cddr{\pi}{n}}$, in order to make some expressions more readable.

\begin{theorem}[Strong normalization]\label{thm:strong-norm}
Let $M \in \KTer{t}$, $\pi \in \KTer{s}$ and $P \in \KTer{p}$.
\begin{itemize}
\item[(i)] If there exist $\Delta,A$ such that $\tystk{s}{\pi}{A}{\Delta}$, then $\pi \in \SN{s}$;
\item[(ii)] If there exist $\Delta,A$ such that $\tystk{t}{M}{A}{\Delta}$, then $M \in \SN{t}$;
\item[(iii)] If there exist $\Delta$ such that $\tystk{p}{P}{}{\Delta}$, then $P \in \SN{p}$.
\end{itemize}
\end{theorem}

\begin{proof}
Let $\Delta = \seq \alpha\asm\seq B$, where $\seq \ga = \ga_1,\ldots,\ga_n$ and $\seq B = B_1,\ldots,B_n$. Let $\cI$ be the $\bS$-interpretation
 sending every atom to $\SN{s}$ and set $\pi_i \ass \cadr{\ga_i}{0}\at\ldots\at\cadr{\ga_i}{\rrank{B_i}-1}\at\cddr{\ga_i}{\rrank{B_i}}$, for each $i=1,\ldots,n$ and
 $\seq \pi = \pi_1,\ldots,\pi_n$. An easy induction on the arity of formulas shows that $\seq \pi \in \fval{\seq B}_{\cI}$.
 By Theorem \ref{thm:adequacy} (i),(ii),(iii) respectively we get that\\
\noindent (i) $\varpi\sub{\seq \pi}{\seq \ga} \in \fval{A}_{\cI} \subseteq \SN{s}$,
          (ii) $M\sub{\seq \pi}{\seq \ga} \in \tval{A}_{\cI} \subseteq \SN{t}$ and
          (iii) $P\sub{\seq \pi}{\seq \ga} \in \SN{p}$.\\
Finally in each of the above cases we have $E\sub{\seq \pi}{\seq \ga} \mslabelto{\eta} E$ and since $E\sub{\seq \pi}{\seq \ga}$ is strongly
 normalizing, then so is $E$.
\end{proof}
\section{Soundness and completeness of typed stack calculus w.r.t. classical semantics}\label{sec:sound-compl-stack}

The present section provides soundness and completeness proofs of the stack calculus for the two-valued semantics of classical propositional logic. We find
 interesting to report the full completeness proof, which resembles very much a completeness proof for a tableaux calculus \cite{Smullyan68}. In fact, as in a
 tableaux system there are labeled formulas (with \emph{true} and \emph{false} labels), in the stack calculus we have terms and stacks which play, respectively,
 the role of proofs and counter-proofs, exactly in the spirit of Krivine's Classical Realizability.

It is easy matter to check that $\bB = (\emptyset,\KTer{t},\KTer{s})$ is a realizability triple. For every formula $A$ and $\bB$-interpretation
 $\cI$ we have
\[
\tval{A}_{\cI} =
\begin{cases}
\KTer{t}        & \text{ if } \fval{A}_{\cI} = \emptyset \\
\emptyset       & \text{ otherwise }
\end{cases}
\]
The induced function $\tval{\cdot}_{\cI}$ maps formulas into elements of the two-element boolean algebra $\{\KTer{t},\emptyset\}$, where the
 ordering is set-inclusion and the operators are $\cup$, $\cap$ and complement. In other words $\KTer{t}$ represents ``true" and $\emptyset$ represents ``false".
 The truth values behave as expected w.r.t. negation: $\tval{A}_{\cI}=\emptyset \iff \tval{\neg A}_{\cI}=\KTer{t}$.
 
\begin{definition}\label{def:sem-entail}
Let $\Phi$ be a set of formulas and let $A$ be a formula. We say that $\Phi$ \emph{semantically entails} $A$, notation $\Phi \vDash A$, if for every atomic
 $\bB$-interpretation $\cI$ we have that $\bigcap_{B \in \Phi} \tval{B}_{\cI} \subseteq \tval{A}_{\cI}$.
\end{definition}

\begin{theorem}[Soundness]\label{thm:soundness}
\ 
\begin{enumerate}[(i)]
\item If $\tystk{t}{M}{A}{\seq \gb \asm \seq B}$ is provable (where $\FV(M) \subseteq \seq \gb$), then
 $\neg B_1,\ldots,\neg B_n \vDash A$.
\item If $\tystk{s}{\pi}{A}{\seq \gb \asm \seq B}$ is provable (where $\FV(\pi) \subseteq \seq \gb$), then
 $\neg B_1,\ldots,\neg B_n \vDash \neg A$.
\item If $\tystk{p}{P}{}{\seq \gb \asm \seq B}$ is provable (where $\FV(P) \subseteq \seq \gb$), then
 $\neg B_1,\ldots,\neg B_n \vDash \bot$.
\end{enumerate}
\end{theorem}

\begin{proof}
\noindent(i) Let $\cI$ be a $\bB$-interpretation. By Theorem \ref{thm:adequacy} (Adequacy) if for all $i \in [1,n]$ $\fval{B_i}_{\cI}\neq\emptyset$,
 then $M\sub{\seq \pi}{\seq \ga} \in \tval{A}_{\cI}$, i.e., $\tval{A}_{\cI} \neq \emptyset$. Since 
$\fval{B_i}_{\cI}\neq\emptyset \iff \tval{B_i}_{\cI} =\emptyset \iff \tval{\neg B_i}_{\cI} =\KTer{t}$, we conclude that every derivable
 judgement $\tystk{t}{M}{A}{\seq \gb \asm \seq B}$ has the following property: for every $\cI$, if $\tval{\neg B_i}_{\cI} =\KTer{t}$ for
 all $i\in [1,n]$, then $\tval{A}_{\cI} = \KTer{t}$. This means, by definition, that $\neg B_1,\ldots, \neg B_n \vDash A$.\\
\noindent(ii),(iii) Similar to (i), again applying Theorem \ref{thm:adequacy}.
\end{proof}

The main goal of the rest of the section is to prove that every classical tautology is the type of some term of the stack-calculus.
 The proof is supported by some auxiliary definitions and lemmas.

\begin{definition}\label{def:higher-terminal-rank}
Let $A$ be a formula. We define its \emph{terminal} $\termin{}{}{A}$ by induction setting $\termin{}{}{G} = G$ and $\termin{}{}{A \to B} = \termin{}{}{B}$.
 We also define its \emph{premisses} $\prem{}{}{A}$ by induction setting $\prem{}{}{G} = \emptyset$ and $\prem{}{}{A \to B} = \{A\} \cup \prem{}{}{B}$.
\end{definition}


\begin{definition}\label{def:higher_premterm}
Let $\Phi$ be a set of formulas. We define three sets $\termin{}{}{\Phi} = \{ \termin{}{}{A} \st A \in \Phi \}$,\\
 $\prem{}{}{\Phi} = \bigcup_{A \in \Phi}  \prem{}{}{A}$, and 
 $\premterm{}{}{\Phi} = \{A \in \prem{}{}{\Phi} \st \termin{}{}{A} \in (\termin{}{}{\Phi} \cup \{\bot\}) \}$.
\end{definition}

\begin{definition}\label{def:higher_complete-set}
A set $\Phi$ of formulas is \emph{saturated} if for every formula $A \in \premterm{}{}{\Phi}$ we have $\prem{}{}{A} \cap \Phi \neq \emptyset$.

\end{definition}

It will turn out that, by applying an iterative process, it is possible to construct saturated sets of formulas starting from finite sets of formulas
 which cannot be proved by a sequent of the stack calculus. The forthcoming Lemmas \ref{lem:higher_unprovable-property-1} and \ref{lem:higher_unprovable-property-2} are the fundamental ingredients for such construction.
 We write $\ntystk{t}{-}{A}{- \asm \seq B}$ to express the fact that there are no variables $\seq \gb$ and no term $M$ such that
 $\tystk{t}{M}{A}{\seq \gb \asm \seq B}$.

\begin{lemma}\label{lem:higher_unprovable-property-1}
Let $\Phi = \{B_0,\ldots,B_n\}$ be a finite set of formulas and suppose $\ntystk{t}{-}{B_0}{- \asm B_1,\ldots,- \asm B_n}$. Then $\premterm{}{}{\Phi} \cap \At = \emptyset$.
\end{lemma}

\begin{proof}
We prove the contrapositive statement. Supposing $A \in \premterm{}{}{\Phi} \cap \At$, we distinguish two possible cases: (1) and (2).
 We write $\seq \gb \asm \seq B$ for the context $\gb_1 \asm B_1,\ldots,\gb_n \asm B_n$. Let $\epsilon$ be a fresh variable.
\begin{enumerate}[(1)]
\item There exist some $j,k \in [0,n]$ such that $B_j =C_1' \to \cdots \to C_i' \to \cdots \to C_{m'}' \to G'$,\\
 $B_k = C_1'' \to \cdots \to C_{m''}'' \to G''$, and $C_i' = G'' = A$. Then
 $\tystk{t}{\bd\gb_0.(\bd\epsilon.\cadr{\gb_j}{i-1}\ap\cddr{\epsilon}{m''})\ap\gb_k}{B_0}{\seq \gb\asm \seq B}$.
\item There exist some $j \in [0,n]$ such that $B_j = C_1' \to \cdots \to C_i' \to \cdots \to C_{m'}' \to G'$, and $C_i' = \bot = A$.
 Then $\tystk{t}{\bd\gb_0.(\bd\epsilon.\cadr{\gb_j}{i-1}\ap\nil)\ap\gb_k}{B_0}{\seq \gb\asm \seq B}$.
\end{enumerate}
\end{proof}

\begin{lemma}\label{lem:higher_unprovable-property-2}
Let $\Phi = \{B_0,\ldots,B_n\}$ be a finite set of formulas and suppose $\ntystk{t}{-}{B_0}{- \asm B_1,\ldots,- \asm B_n}$. Then for every $A \in \premterm{}{}{\Phi}$ there exists a formula $C \in \prem{}{}{A}$ such that \mbox{$\ntystk{t}{-}{B_0}{- \asm B_1,\ldots,- \asm B_n,- \asm C}$}.
\end{lemma}

\begin{proof}
We prove the contrapositive statement. To this end, suppose $A \in \premterm{}{}{\Phi}$ is a formula that is a counterexample to the conclusion of the statement. First note that $\prem{}{}{A} \neq \emptyset$, otherwise $A \in \At$, in contradiction with Lemma \ref{lem:higher_unprovable-property-1}.
 Therefore $A = C_1 \to \cdots \to C_m \to G$, with $m\geq 1$. We write $\seq \gb \asm \seq B$ for the context $\gb_1 \asm B_1,\ldots,\gb_n \asm B_n$.

By our assumption for every $i=1,\ldots,m$ ($m \geq 1$) there exist $M_i,\gc_i$ such that $\tystk{t}{M_i}{B_0}{\seq \gb \asm \seq B,\gc_i \asm C_i}$
 and thus we derive $\tystk{t}{\bd\gc_i.M_i\ap\gb_0}{C_i}{\gb_0 \asm B_0,\seq \gb \asm \seq B}$ for each $i=1,\ldots,m$.
 Moreover, since $A \in \premterm{}{}{\Phi}$, there are two cases:
\begin{enumerate}[(1)]
\item there exist some $k,h \in [0,n]$ such that $B_h = {C}_1' \to \cdots \to {C}_j' \to \cdots \to {C}_{m'}' \to G'$,\\
 $B_k = {C}_1'' \to \cdots \to {C}_{m''}'' \to G''$,
 $A = {C}_j'$, and $G = G''$.
\item $G = \bot$ and there exist some $h \in [0,n]$ such that $B_h = {C}_1' \to \cdots \to {C}_j' \to \cdots \to {C}_{m'}' \to G'$ and $A = {C}_j'$.
\end{enumerate}

Let $\epsilon$ be a fresh variable. In both cases (1) and (2) there exists a stack $\pi$ such that $\tystk{s}{\pi}{G}{\epsilon\asm B_k}$
 is derivable, where $\pi$ is either $\nil$ or $\cddr{\epsilon}{\rrank{B_k}}$.

Let $\gc_1,\ldots,\gc_m,\gd$ be fresh variables and let $\varpi \ass (\bd\gc_1.M_1\ap\gb_0)\at\ldots\at(\bd\gc_m.M_m\ap\gb_0)\at\pi$.
 Then we finally derive $\tystk{t}{\bd\gb_0.(\bd\gd.(\bd\epsilon.\cadr{\gd}{j-1}\ap\varpi)\ap\gb_k)\ap \gb_h}{B_0}{\seq \gb \asm \seq B}$.
\end{proof}

The \emph{complexity} of a formula $A$ is the total number of implications and atomic sub-formulas occurring in $A$. The formulas of complexity one
 are exactly the atomic ones.

\begin{lemma}\label{lem:higher_sat-false-interpretation}
Let $\Phi$ be a saturated set of formulas. Then there exists a $\bB$-interpretation $\cI$ such that $\tval{A}_{\cI}= \emptyset$, for all $A \in \Phi$.
\end{lemma}

\begin{proof}
The case in which $\Phi = \emptyset$ is trivial, so for the rest of the proof we assume $\Phi \neq \emptyset$. We define a
 $\bB$-interpretation $\cI$ as follows:
\[
\cI(G) =
\begin{cases}
\emptyset  & \text{ if } G \in \termin{}{}{\Phi} \\
\KTer{t}   & \text{ otherwise } \\
\end{cases}
\]
We now prove that $\tval{A}_{\cI} = \emptyset$, for all $A \in \Phi$. The proof is by induction on the complexity of formulas.

\noindent Suppose $A \in \At$. If $A=\bot$ the result is obvious; otherwise, since $A  \in \termin{}{}{\Phi}$, we have $\tval{A}_{\cI} = \emptyset$.

\noindent Suppose $A = C_1 \to \cdots \to C_m \to G$ (with $m\geq 1$). We now prove that \\
\begin{itemize}
\item[ ] (1)\quad $\tval{C_1}_{\cI}= \cdots = \tval{C_m}_{\cI} = \KTer{t}$; \qquad \qquad \qquad (2)\quad $\tval{G}_{\cI} = \emptyset$.
\end{itemize}
The items (1) and (2) together yield $\tval{A}_\cI = \emptyset$.
\begin{enumerate}[(1)]
\item For $C_i \in \prem{}{}{A}$ we distinguish two cases.\\
\noindent Suppose $C_i \not\in \premterm{}{}{\Phi}$. Then $\termin{}{}{C_i}$ is not a terminal of a formula in $\Phi$. By definition of
 $\cI$ we have $\tval{\termin{}{}{C_i}}_{\cI} = \KTer{t}$. We conclude observing that
 $\tval{C_i}_{\cI} \supseteq \tval{\termin{}{}{C_i}}_{\cI} = \KTer{t}$.\\
\noindent Suppose $C_i \in \premterm{}{}{\Phi}$. Then, by saturation of $\Phi$, $C_i = C_1' \to \cdots \to C_{m'}' \to G'$ (with $m'\geq 1$) and there
 exists $j \in [1,m']$ such that $C_j' \in \Phi$. Since $C_j'$ has strictly lower complexity than $A$, by induction hypothesis
 $\tval{C_j'}_{\cI} = \emptyset$. This implies $\tval{C_i}_{\cI}= \KTer{t}$.
\item Since $G \in \termin{}{}{\Phi} \cup \{\bot\}$, evidently $\tval{G}_{\cI} = \emptyset$ by the definition of the interpretation $\tval{\cdot}_{\cI}$.
\end{enumerate}
\end{proof}

Next we give the second main theorem of this section, concerning completeness.
 The idea of its proof is the \emph{counter-model construction}, typical of Smullyan's analytic tableaux \cite{Smullyan68}.

\begin{theorem}[Completeness]\label{thm:higher-completeness}
Let $A$ be a formula and let $\seq B$ be a sequence of formulas. If $\neg B_1, \ldots, \neg B_n \vDash A$, then there exist $M$ and
 $\seq \gb$ such that $\tystk{t}{M}{A}{\seq \gb \asm \seq B}$ is provable.
\end{theorem}

\begin{proof}
We proceed to prove the contrapositive statement. Suppose $\ntystk{t}{-}{A}{- \asm \seq B}$. Then we can construct a saturated set $\Phi$ of formulas
 containing $\{A,B_1,\ldots,B_n\}$ as $\Phi \ass \bigcup_{n \geq 0} \Phi_n$, where the family $\{\Phi_n\}_{n \geq 0}$ is
 inductively defined as follows:
\begin{itemize}
\item $\Phi_0 \ass \{A,B_1,\ldots,B_n\}$;
\item If $\premterm{}{}{\Phi_{n}} = \emptyset$, then we define $\Phi_{n+1} \ass \Phi_{n}$.
 If $\premterm{}{}{\Phi_{n}} = \{C_1,\ldots,C_k\} \neq \emptyset$, by Lemma \ref{lem:higher_unprovable-property-2} for each $C_i$ there
 exists a formula $D_i \in \prem{}{}{C_i}$ such that \mbox{$\ntystk{t}{-}{A}{- \asm B_1,\ldots,- \asm B_n,- \asm D_i}$}.
 Let $\Psi_n = \{D_1,\ldots,D_k\}$, where each $D_i$ is the leftmost premiss of $C_i$ having the property that
 $\ntystk{t}{-}{A}{- \asm B_1,\ldots,- \asm B_n,- \asm D_i}$. Then we define $\Phi_{n+1} \ass \Phi_{n} \cup \Psi_n$.
\end{itemize}
By construction $\Phi$ is a saturated set of formulas containing $\{A,B_1,\ldots,B_n\}$. Finally applying Lemma \ref{lem:higher_sat-false-interpretation} we obtain some $\cI$ such that
 $\tval{B_1}_{\cI} =\cdots= \tval{B_n}_{\cI}= \tval{A}_{\cI}=\emptyset$, meaning that $\neg B_1, \ldots,\neg B_n \nvDash A$.
\end{proof}

Of course Theorem \ref{thm:higher-completeness} implies that every classical propositional tautology (of the $\{\to,\bot\}$-fragment) is provable
 by the type derivation of a term.

\section{The Krivine machine for stack calculus}\label{sec:Krivine-machine}

In the present section we sketch the definition of a Krivine machine that executes the terms of stack calculus.
 Similar machines have been defined by de Groote \cite{Groote98}, Laurent \cite{Laurent03note}, Reus and Streicher
 \cite{Stre98} for the \lmu-calculus. Using this machine we show how to encode control mechanisms like \emph{label/resume} and \emph{raise/handle}
 in the stack calculus.

In order to define the states of the machine, we need the following mutually inductive definitions.
 A \emph{stack closure} is a pair $p=\Ksclos{\pi}{e}$ consisting of a stack $\pi$ and an environment $e$;
 a \emph{term closure} is a pair $m=\Ktclos{M}{e}$ consisting of a stack $\pi$ and an environment $e$;
 an \emph{environment} is a partial function (with finite domain) from the set of stack variables to the set of stack closures.
 We write $\envup{e}{\ga}{p}$ for the environment $e'$ which assumes the same values as $e$ except at most on $\ga$, where $e'(\alpha)=p$.

A \emph{state} is a pair $\Kstate{m}{p}$ and the machine consists of the following (deterministic) transitions between states:
\[
\begin{array}{lcll}
\Kstate{\Ktclos{N}{e}}{p}  & \transition & \Kstate{\Ksclos{\cadr{\pi'}{n}}{e'}}{p}
       & \text{ if } \cadr{\ga}{n}  \text{ is the } \mslabelto{\mathsf{car},\mathsf{cdr}} \text{-normal form of } N \text{ and } e(\ga) = \Ksclos{\pi'}{e'} \\
\Kstate{\Ktclos{N}{e}}{p} & \transition & \Kstate{\Ktclos{M}{e'}}{\Ksclos{\pi}{e'}} 
       & \text{ if } \bd\ga.M\ap\pi \text{ is the } \mslabelto{\mathsf{car},\mathsf{cdr}} \text{-normal form of } N \text{ and } e' = \envup{e}{\ga}{p}
\end{array}  
\]
We let $\transitions$ be the reflexive and transitive closure of the relation $\transition$. Consider a state $\Kstate{\Ktclos{M}{e}}{p}$. The closure $p$ is the
 current context of evaluation of $M$; the next state may discard $p$ and restore a context appeared in the past. The environment $e$ is the current state of the
 memory: it takes into account all side effects caused by the previous stages of computation. The term $M$ is said to be in \emph{execution position} and it is the
 current program acting on $p$ evaluated in $e$. A \emph{computation} is a sequence of states sequentially related by the transition rules.

To explain how the stack calculus achieves the control of the execution flow, we define label/resume and raise/handle instructions and show that the machine soundly
 executes them. We set
\[
\begin{array}{lr}
\lab{\epsilon}{M} \ass \bd\gb.(\bd{\epsilon}.M \ap \gb)\ap(\bd{\gd}.\cadr{\gd}{0}\ap\gb) \at \gb & \text{ with } \gb \not\in \FV(M) \\
\res{\epsilon}{M} \ass \bd\gc.\cadr{\epsilon}{0}\ap N\at\gc & \text{ with } \epsilon,\gc \not\in \FV(M) \\
\delthrow{\epsilon}{M} \ass \bd\gc.\cadr{\epsilon}{0}\ap M\at\nil & \text{ with } \epsilon,\gc \not\in \FV(M) \\
\deltrycatch{\epsilon}{M}{N} \ass \bd\gb.(\bd\epsilon.M\ap\gb)\ap (\bd\gd.N\ap\cadr{\gd}{0}\at\gb)\at\nil &  \text{ with } \gb\not\in (\FV(M) \cup \FV(N)),\ \gd \not \in \FV(N) 
\end{array}
\]

We now discuss briefly and informally how the machine executes the above instructions.

Suppose to start the machine in a state $S = \Kstate{\Ktclos{\lab{\epsilon}{M}}{e_0}}{p_0}$. If no term $\res{\epsilon}{N}$ ever reaches the execution
 position, then the computation starting at $S$ is equivalent to that starting at $S' = \Kstate{\Ktclos{M}{e_0}}{p_0}$.
 Otherwise $S \transitions^n \Kstate{\Ktclos{\bd\gc.\cadr{\epsilon}{0}\ap N\at\gc}{e_n}}{p_n} \transitions^2 \Kstate{\Ktclos{N}{e_{n+1}}}{p_{n+2}}$,
 and we notice that the computation starting at $\Kstate{\Ktclos{\res{\epsilon}{N}}{e_n}}{p_n}$ is equivalent to that starting at $\Kstate{\Ktclos{N}{e_{n+1}}}{p_{n+2}}$.

Suppose to start the machine in a state $S = \Kstate{\Ktclos{\deltrycatch{\epsilon}{M}{N}}{e_0}}{p_0}$. If no term $\delthrow{\epsilon}{M'}$ ever
 reaches the execution position, then the computation starting at $S$ is equivalent to that starting at $S' = \Kstate{\Ktclos{M}{e_0}}{p_0}$.
 Otherwise $S \transitions^n \Kstate{\Ktclos{\bd\gc.\cadr{\epsilon}{0}\ap M'\at\nil}{e_n}}{p_n} \transitions^3 \Kstate{\Ktclos{N}{e_{n+2}}}{\Ksclos{\cadr{\gd}{0}\at\gb}{e_{n+2}}}$
 and we can see that the exception handler $N$ goes on with the computation, and the value $M'$ returned by the exception is at use of $N$, since it is stored in the
 in the current environment $e_{n+2}$ in a cell that is present in the current evaluation context.

We conclude remarking that all the above constructions can be typed by derived rules. Informally one may assert that Theorem \ref{thm:soundness} and Theorem \ref{thm:subj-red-stk},
 together, ensure that the execution of well-typed term always ensures that all the ``resume'' and ``raise'' instructions are always handled correctly.

\section{Denotational semantics of stack calculus}\label{sec:den-sem}

Girard's \emph{correlation spaces} \cite{Gir91} are (one of) the first denotational models of Classical Logic: they refine
 coherence spaces \cite{Gir86} with some additional structure. Intuitively, these richer objects come with the information required to interpret structural
 rules (weakening and contraction) on the right-hand side of sequents in classical sequent calculus. Girard's construction
 hints that Classical Logic may be encoded into Linear Logic, a result achieved by Danos et al. \cite{DanosJS95}
 via a dual linear decomposition of classical implication. In \cite{Stre98} the authors interpret the \lmu-calculus in the Cartesian closed category of ``negated domains'', i.e. the full subcategory of $\CPO$
 determined by the objects of the form $R^A$, where $A$ is a predomain and $R$ is some fixed domains of ``responses''.
 The category of negated domains is a particular \emph{category of continuations} \cite{Stre93} and categories
 of continuations are \emph{complete} \cite{Stre97} for the \lmu-calculus, in the sense that every equational theory for 
 \lmu-calculus is given by the kernel relation of the interpretation in some category of continuations.
 Selinger \cite{Selinger01} gives a general presentation in terms of \emph{control categories}, which are easily seen to subsume
 categories of continuations. However via a categorical structure theorem he also shows that every control
 category is equivalent to a category of continuations. This structure theorem implies the soundness and completeness
 of the categorical interpretation of the \lmu-calculus with respect to a natural CPS semantics.

In brief, a control category is a Cartesian closed category $(\bC,\wi,\top,\Rightarrow)$ which is also a symmetric premonoidal category
 $(\bC,\pa,\bot)$. The binoidal functor $\pa$ distributes over $\wi$ and there is a natural isomorphism $s_{A,B,C}: B^A \pa C \to (B \pa C)^A$ in $A,B$ and $C$ satisfying some
 coherence conditions. Selinger distinguishes a subcategory $\foc{\bC}$ of $\bC$, called the \emph{focus} of $\bC$, which have the same objects as
 $\bC$ but fewer arrows. On $\foc{\bC}$ the functor $\pa$ restricts to a coproduct. It is very important to remark that in any control category $\bC$
 there exists an isomorphism $\ctrliso:\bC(\top,B\pa A) \cong \foc{\bC}(\bot^A,B)$ natural in central $B$ (see \cite{Selinger01} for the details).
 If $\bC$ is a control category we map falsity to the object $\bot$ and set $\tval{A \to B} = \bot^{\tval{A}} \pa \tval{B}$; a context $\Delta = \seq \ga\asm\seq A$ is mapped to
 $\tval{\Delta} = \tval{A_1} \pa \cdots \pa \tval{A_n}$. Then the judgements are interpreted as morphisms
 $\Int{\tystk{s}{\pi}{A}{\Delta}}: \tval{A} \to \tval{\Delta}$,
 $\Int{\tystk{t}{M}{A}{\Delta}}:   \bot^{\tval{A}} \to \tval{\Delta}$ and
 $\Int{\tystk{p}{P}{}{\Delta}}:    \top \to \tval{\Delta}$, using the coproduct structure and the isomorphism $\ctrliso$.
 The above intepretation is \emph{sound}, in the sense that it is invariant under $\labelto{\s\eta}$-reduction of expressions.

Very interesting is the work of Laurent and Regnier \cite{Laurent03} which shows in detail how to extract a control category
 out of a categorical model of MALL. This constribution gives a general framework under which falls the correlation spaces model construction by Girard and
 at the same time constitutes the categorical counterpart of Danos\textendash Joinet\textendash Schellinx's \cite{DanosJS95} call-by-name encoding
 of Classical logic into Linear Logic.

A $\ast$-autonomous category is a symmetric monoidal category with two monoidal structures $(\bC,\otimes,\teid)$ and $(\bC,\pa,\bot)$ possessing
 a \emph{dualizing} endofunctor $\ort{(\cdot)}$ which maps $f:A\to B$ to $\ort f: \ort B \to \ort A$.

Let $\bC$ be a $\ast$-autonomous category. When the forgetful functor from the category $\Mon{\bC}$ (of $\pa$-monoids and $\pa$-monoid morphisms)
 to the category $\bC$ has a right adjoint, then $\bC$ is a \emph{Lafont category}. We recall that the co-Kleisli category $\mbbK_{\bC}$ of a monoidal
 category $\bC$ via a comonad $(!,\dig,\der)$ has the same objects as $\bC$ and $\mbbK_{\bC}(A,B) = \bC(!A,B)$; the composition of morphisms
 is defined using the monad structure (see \cite{Mellies09}).
 One of the main results of \cite{Laurent03} is that if $\bC$ is a $\ast$-autonomous Lafont category with finite products, then then the co-Kleisli category
 $\mbbK_{\bC'}$ of the full-subcategory $\bC'$ of $\bC$ whose objects are the $\pa$-monoids is a control category.

\subsection{A simple interpretation of stack calculus}

Inspired by Laurent and Regnier's work \cite{Laurent03} we give a minimal framework in which the stack calculus can be soundly interpreted.
 The absence of the $\lambda$-abstraction, allows us to focus on the minimal structure required to interpret Laurent's Polarized
 Linear Logic \cite{Laurent03b} and to use it to interpret the stack calculus.

Let $\bC$ be a $\ast$-autonomous category. We denote by $\rho_A: A \to A \pa \bot$, $\lambda_A: A \to \bot \pa A$, $\alpha$, $\gamma$ and $\tau$ the usual natural isomorphisms related to
 the monoidal structure of $(\bC,\pa,\bot)$. 

A \emph{linear category} is a symmetric monoidal category together with a symmetric monoidal comonad $((\bang,\mon),\dig,\der)$ such that there are
 monoidal natural transformations with components $\bcouni_A:\ \bang A \to \bon$ and $\bcomul_A:\ \bang A \to \bang A \otimes \bang A$
 which are coalgebra morphisms and make each free $\bang$-coalgebra a commutative $\otimes$-comonoid $(\bang A,\comul_A,\couni_A)$; moreover
 $\dig_A:\ \bang A \to \bang\bang A$ is a comonoid morphism, for every object $A$.


In the sequel we let $\bC$ be a $\ast$-autonomous linear category, so that by duality we can turn the above definition in terms of
 a monad $((?,\mon),\dig,\der)$, $?$-algebras and $\pa$-monoids. In this case there are monoidal natural transformations with components $\uni_A: \bot \to ?A$ and $\mul_A: ?A\pa ?A \to ?A$
 which are $?$-algebra morphisms and make each free $?$-algebra a commutative $\pa$-monoid $(?A,\bmul_{A},\buni_{A})$;
 $\dig_A: ??A \to ?A$ is a monoid morphism, for every object $A$.
 Under these hypotheses all $?$-algebras $A$, being retract of a the free algebra $?A$, have a multiplication $\mul_{A}$, and a unit $\uni_{A}$
 (see \cite{Mellies09} for further details). The category $\bC^?$ of Eilenberg-Moore algebras is symmetric monoidal, with (co)tensor product of $(A,\alg_A)$, $(B,\alg_B)$ given by
 $(A \pa B,(\alg_A\pa\alg_B)\circ \mon^2)$ and unit given by $(\bot,\mon^1)$.
 The $\ast$-autonomous structure of $\bC$ yields a natural isomorphism $\cur:\bC(\bon,B \pa A) \to \bC(\ort{A},B)$ that we will use to interpret abstraction
 (a natural retraction $\bC(\bon,B \pa A) \triangleleft \bC(\ort{A},B)$ would suffice anyway).

Starting from a valuation that associates $?$-algebras to atomic types and the object $\bot$ to falsity, the arrow-types are mapped as follows:
 $\tval{A \to B} = ?\ort{\tval{A}} \pa \tval{B}$. Given a context $\Delta = \seq \ga\asm\seq A$ we set $\tval{\Delta} = \tval{A_1} \pa \cdots \pa \tval{A_n}$.
 Note that all types are interpreted by $?$-algebras. Then the type judgements with assumptions $\Delta$ can be easily interpreted as morphisms with 
 target $\tval{\Delta}$; for example $\Int{\tystk{s}{\nil}{\bot}{\Delta}}: \bot \to \tval{\Delta}$ is the unit of the monoid $\tval{\Delta}$.
 We describe such interpretation for the particular case of the untyped stack calculus, for which we need a $?$-algebra $U$ of
 $\bC$ together with two $?$-algebra morphisms $\lam: ?\ort{U} \pa U \to U$ and $\app: U \to ?\ort{U} \pa U$ satisfying $\app \circ \lam = \id_{?\ort{U} \pa U}$
 and a $?$-algebra morphism $\vartheta: U \to \bot$ (needed for the stack $\nil$).
 We write $U^n$ for the $n$-fold $\pa$-product of $U$. Such product inherits a $?$-algebra structure $\alg_{U^n}$ defined using the algebra
 $\alg_U$ and the monoidality of the monad; as a consequence it also inherits a multiplication $\mul_{U^n}$ and a unit $\uni_{U^n}$. We also define
 ${\iota_j}^{n}: U \cong \bot^{j-1} \pa U \pa \bot^{n-j} \xrightarrow{\uni_{U^{j-1}}\pa\id_{U}\pa\uni_{U^{n-j}}} U^n$.

For all expressions $E$ with $\FV(E) \subseteq \seq \ga$ we define the interpretation
 $\Int{M}_{\seq \ga}: \ort{U} \to U^n$, $\Int{\pi}_{\seq \ga}: U \to U^n$ and $\Int{P}_{\seq \ga}: \bon \to U^n$
 as follows ($n=\leng{\seq \ga}$):
\[
\begin{array}{lll}
\Int{\ga_i}_{\seq \ga} = {\iota_j}^{n}
 & 
\Int{M\at\pi}_{\seq \ga} = \copair{\alg_{U^n} \circ ?\Int{M}_{\seq \ga}}{\Int{\pi}_{\seq \ga}} \circ \app
 &
\Int{\cdr\pi}_{\seq \ga} = \Int{\pi}_{\seq \ga} \circ \lam \circ (\uni_{\ort{U}}\pa\id_{U})\circ \rho_{U}
 \\ & & \\
\Int{\nil}_{\seq \ga} = \uni_{U^n} \circ \vartheta
 & 
\Int{\car\pi}_{\seq \ga} = \Int{\pi}_{\seq \ga} \circ \lam \circ (\der_{\ort{U}}\pa\uni_{U})\circ \lambda_{\ort{U}}
 & 
\Int{\bd{\gb}.P}_{\seq \ga} = \cur(\Int{P}_{\seq \ga,\gb})
 \\ & & \\
 & 
\Int{M\ap\pi}_{\seq \ga} = \copair{\id_{U^n}}{\Int{\pi}_{\seq \ga}}\circ\cur^{-1}(\Int{M}_{\seq \ga})
 & 
\end{array}
\]

Note that the denotations of stacks are $?$-algebra morphisms and it is not difficult to verify that the above interpretation is invariant under
 $\labelto{\s}$-reduction. To see that check before that $\Int{E\sub{\pi}{\gb}}_{\seq \ga} = \copair{\id_{\seq\ga}}{\Int{\pi}_{\seq\ga}} \circ \Int{E}_{\seq\ga,\gb}$.
 The category $\Rel$ of sets and relations is a $\ast$-autonomous linear category that satisfies all our requirements \cite{Mellies09}.
 If $S$ is a set, we denote by $\Omegatuple{S}$ the set of all the $\nat$-indexed sequences $\gs = (a_1,a_2,\dots)$ of multisets over $S$ such that
 $a_i = []$ holds for all but a finite number of indices $i\in \nat$.
 The set $\Omegatuple{S}$ is a simple example of $?$-algebra of $\Rel$. For $\gs = (a_1,a_2,\dots)$ and $\gt = (b_1,b_2,\dots)$, we define
 $\gs + \gt = (a_1 \mcup b_1,a_2\mcup b_2,\dots)$
 and $\ast = ([],[],\dots)$. Then the relations $\uni = \{(1,\ast)\}$ and $\mul = \{((\gs,\gt),\gs+\gt) \st \gs,\gt \in \Omegatuple{S}\}$
 make $(\Omegatuple{S},\mul,\uni)$ a $\pa$-monoid in $\Rel$. The operation $+$ on $\Omegatuple{S}$ can also be extended componentwise to $(\Omegatuple{S})^k$
 (whose elements are ranged over by $\seq \gs,\seq\gt,\ldots$) transferring thereby the monoid structure.
 In order to model the untyped calculus we need a $\pa$-monoid $U$ of together with two relations $\lam \subseteq (\Mfin{U} \times U) \times U$ and
 $\app \subseteq U \times (\Mfin{U} \times U)$ satisfying $\app \circ \lam = \id_{\Mfin{U} \times U}$ and a relation $\vartheta \subseteq U \times \{1\}$.
 In the category $\Rel$ lives one such object $\cD=(D,\app,\lam)$ that has already been encountered many times in the literature (see for example
 \cite{BucciarelliEM07}) as a model of the ordinary $\lambda$-calculus (as well as of some of its extensions). The object is constructed as union $D=\bigcup_{n\in\nat}D_n$ of a family of sets $(D_n)_{n\in\nat}$
 defined by $D_0=\emptyset$ and $D_{n+1}=\Omegatuple{D_n}$. Given $\gs=(a_1,a_2,a_3,\ldots)\in D$ and $a\in\Mfin{D}$, we write $a\cons\gs$ for the
 element $(a,a_1,a_2,a_3,\ldots)\in D$. Since $D = \Omegatuple{D}$, as previously observed it has a standard monoid structure
 and we can set $\lam = \{((a,\gs),a\cons\gs) \st a \in \Mfin{D},\ \gs \in D\}$ and $\app = \{(a\cons \gs,(a,\gs)) \st a \in \Mfin{D},\ \gs \in D\}$ satisfying
 the desired equation; as a matter of fact also the equation $\lam \circ \app = \id_{U}$ holds and the interpretation
 of expressions is invariant under $\labelto{\s\eta}$-reduction. Finally $\vartheta = \{(\ast,1)\}$.

The isomorphism $\cur:\bC(\bon,U \pa U) \to \bC(\ort{U},U)$ is trivially given by $\cur(f) = \{(\ga,\gb) \st (1,(\gb,\ga)) \in f\}$.
 The interpretation is concretely defined as follows:
\[
\begin{array}{l}
\Int{\ga_i}_{\seq \ga} = \{(\gs,(\ast \ldts \gs \ldts \ast)) \st \gs \in D \}; \qquad \qquad \qquad \qquad \quad
\Int{\cdr\pi}_{\seq \ga} = \{(\gs,\seq \gt) \st ([]\cons\gs,\seq \gt) \in \Int{\pi}_{\seq \ga} \}; \\
\Int{\car\pi}_{\seq \ga} = \{(\gs,\seq \gt) \st ([\gs]\cons\ast,\seq \gt) \in \Int{\pi}_{\seq \ga} \}; \qquad \qquad \qquad \quad
\Int{\bd{\gb}.P}_{\seq \ga} = \{(\gs,\seq \gt) \st (1,(\seq \gt,\gs)) \in\Int{P}_{\seq \ga,\gb} \}; \\
\Int{M\at\pi}_{\seq \ga} = \{([\gs_1\ldts\gs_k]\cons\gs,\Sigma_{i=0}^{k}\seq \gt_i) \st k\geq 0,\forall i=1\ldts k.\ (\gs_i,\seq \gt_i) \in \Int{M}_{\seq \ga},\ (\gs,\seq \gt_0) \in \Int{\pi}_{\seq \ga} \}; \\
\Int{M\ap\pi}_{\seq \ga} = \{(1,\seq\gt+\seq\gt') \st \exists \gs \in D.\ (\gs,\seq \gt) \in \Int{M}_{\seq \ga},\ (\gs,\seq \gt') \in \Int{\pi}_{\seq \ga}\};\qquad
\Int{\nil}_{\seq \ga} = \{(\ast,(\ast \ldts \ast))\}.
\end{array}
\]


For example for the stack calculus version of $\mathsf{call}/\mathsf{cc}$ we have\\
$\Int{\mu\ga.\cadr{\ga}{0}\ap(\mu\gb.\cadr{\gb}{0}\ap\cddr{\ga}{1})\at\cddr{\ga}{1}}=
\{ [[[\gs_1]\cons\ast\ldts[\gs_k]\cons\ast]\cons\gs_0]\cons(\Sigma_{i=0}^k\gs_i) \st k\geq 0,\ \gs_0\ldts\gs_k \in D\}$.

\section{Conclusions}\label{sec:conclusion}

We introduced the stack calculus, a finitary functional calculus with simple syntax and rewrite rules in which the calculi introduced so far in the Curry\textendash Howard correspondence for classical
 logic can be faithfully encoded; instead of exhibiting comparisons with all the existing formalisms, we just showed how Parigot's \lmu-calculus can
 be translated into our calculus. We proved that the untyped stack calculus enjoys confluence, and that types enforce strong normalization. The typed fragment is a sound and complete
 system for full implicational Classical Logic. The type system that Lafont et al. \cite{Stre93} use for the $\lambda$-calculus with pairs
 may be used to type stack expressions within the $\{\wedge,\neg,\bot\}$-fragment of Intuitionistic Logic: under this point of view,
 the stack calculus is the target-language of a CPS translation from itself that alters the types but not the expressions of the calculus.
 In the classically-typed system ($\{\to,\bot\}$-fragment of Classical Logic) the arrow type corresponds to the
 stack constructor; for this reason the realizability interpretation of types \`{a} la Krivine matches perfectly the logical meaning of the arrow in
 the type system. The proofs of soundness and strong normalization of the calculus are both given by particular realizability interpretations. We defined a Krivine machine that executes the terms of stack calculus. We showed how to encode control mechanisms like \emph{label/resume} and
 \emph{raise/handle} in the stack calculus which are soundly executed by our machine. This approach seems to be simpler than the extension of ML
 with exceptions studied in De Groote \cite{Groote95}. Inspired by Laurent and Regnier's work \cite{Laurent03}, we give a simple categorical framework to interpret the expressions of both typed
 and untyped stack calculus. We show how, in the case of a relational semantics, this famework allows a simple calculation of the interpretation
 of expressions.


\bibliographystyle{eptcs}
\bibliography{bibliography}
\end{document}